\newif\ifLIPICS
\LIPICSfalse

\newif\ifanonym %
\anonymtrue
\anonymfalse %

\ifLIPICS
     \documentclass[anonymous, english,cleveref,numberwithinsect]{socg-lipics-v2021}
\else
    \documentclass[11pt,a4]{article}
    \usepackage[margin=1in]{geometry}
    \usepackage[T1]{fontenc}
    \usepackage{babel}
    \usepackage{float}
    \usepackage{amsmath}
    \usepackage{amsthm}
    \usepackage{amssymb}
    \usepackage[unicode=true,pdfusetitle,colorlinks=true,allcolors=blue]{hyperref}
    \usepackage[capitalise,noabbrev]{cleveref}
\fi

\usepackage{xcolor}
\usepackage{graphicx} %

\ifLIPICS\else
\theoremstyle{plain}
\newtheorem{theorem}{Theorem}[section]

\newtheorem{lemma}[theorem]{Lemma}

\newtheorem*{claim*}{Claim}
\newtheorem{corollary}[theorem]{Corollary}
\newtheorem{proposition}[theorem]{Proposition}

\theoremstyle{definition}
\newtheorem{definition}[theorem]{Definition}

\theoremstyle{remark}
\newtheorem{remark}[theorem]{Remark}

\fi

\newtheorem{question}[theorem]{Question}

\DeclareMathOperator{\poly}{poly}

\DeclareMathOperator{\ddim}{ddim}
\DeclareMathOperator{\diam}{diam}
\DeclareMathOperator{\supp}{supp}
\newcommand{\tO}{\tilde{O}}

\newcommand{\RR}{\mathbb{R}}
\newcommand{\cX}{\mathcal{X}}
\newcommand{\cA}{\mathcal{A}}
\newcommand{\Pinit}{P_{\text{init}}}
\newcommand{\Pout}{P_{\text{out}}}

\def\compactify{\itemsep=0pt \topsep=0pt \partopsep=0pt \parsep=0pt}

\title{Lipschitz Decompositions of Finite $\ell_{p}$ Metrics}

\ifLIPICS
    \author{Robert Krauthgamer}{Weizmann Institute of Science}{robert.krauthgamer@weizmann.ac.il}{https://orcid.org/0009-0003-8154-3735}
        {Work partially supported by the Israel Science Foundation grant \#1336/23,
          by the Israeli Council for Higher Education (CHE) via the Weizmann Data Science Research Center,
          and by a research grant from the Estate of Harry Schutzman.}

    \author{Nir Petruschka}{Weizmann Institute of Science}{nir.petruschka@weizmann.ac.il}{}{}

    \authorrunning{R. Krauthgamer and N. Petruschka} %

\else %
    
\ifanonym
\author{Anonymous Authors}
\else
\author{Robert Krauthgamer%
  \thanks{Work partially supported by the Israel Science Foundation grant \#1336/23,
    by the Israeli Council for Higher Education (CHE) via the Weizmann Data Science Research Center,
    and by a research grant from the Estate of Harry Schutzman.
    Part of this work was done while visiting the Simons Institute for the Theory of Computing.
    Email: \texttt{robert.krauthgamer@weizmann.ac.il}
  } 
  \qquad 
Nir Petruschka\thanks{
    Email: \texttt{nir.petruschka@weizmann.ac.il}
  } 
\\  Weizmann Institute of Science
}
\fi %

\fi %

\begin{document}

\maketitle

\begin{abstract}
Lipschitz decomposition is a useful tool in the design of efficient algorithms
involving metric spaces.
While many bounds are known for different families of finite metrics,
the optimal parameters for $n$-point subsets of $\ell_p$, for $p > 2$, 
remained open, see e.g.\ [Naor, SODA 2017]. 
We make significant progress on this question 
and establish the bound $\beta=O(\log^{1-1/p} n)$.
Building on prior work, we demonstrate applications of this result to two problems, 
high-dimensional geometric spanners and distance labeling schemes.
In addition, we sharpen a related decomposition bound for $1<p<2$,
due to Filtser and Neiman [Algorithmica 2022].
\end{abstract}

\section{Introduction}
The pursuit of approximating metric spaces by simpler structures has inspired
the development of fundamental concepts,
such as graph spanners~\cite{PU89a, PS89}
and low-distortion embeddings into various spaces~\cite{LLR95, Bartal96},
both of which have a wide range of algorithmic applications.
Many of these results, including for instance \cite{Bartal96, Rao99, FRT04, KLMN05, ALN08, Filtser24},
rely on various notions of decomposition of a metric space into low-diameter clusters,
and these decompositions are most often randomized.
One extensively studied notion,
see e.g.~\cite{CCGGP98, GKL03, FRT04, HIS13},
is \emph{Lipschitz decomposition} (also called \emph{separating decomposition}),
which informally is a random partition of a metric space into low-diameter clusters,
with a guarantee that nearby points are likely to belong to the same cluster. 

\begin{definition}
[Lipschitz decomposition \cite{Bartal96}]
\label{defn:(Lipschitz-Decomposition)}
Let $(X, \rho)$ be a metric space.
A distribution $\mathcal{D}$ over partitions of $X$ is called \emph{$(\beta, \Delta)$-Lipschitz} if 
\begin{enumerate} \compactify
\item for every partition $P \in \supp(\mathcal{D})$, all clusters $C \in P$
  satisfy $\diam(C) \leq \Delta$; and 
\item for all $x, y \in X$,
  \[
    \Pr_{P \in \mathcal{D}} [P(x) \neq P(y)]
    \leq \beta\cdot \tfrac{\rho(x, y)}{\Delta} ,
  \]
  where $P(z)$ denotes the cluster of $P$ containing $z \in X$
  and $\diam(C) := \sup_{x,y\in C} \rho(x,y)$.
\end{enumerate}
\end{definition}

Typical applications require such decompositions 
where $\Delta$ is not known in advance,
or even multiple values of $\Delta$ (say for every power of $2$). 
We naturally seek small $\beta$ and thus define
the (optimal) \emph{decomposition parameter of $(X,\rho)$} as 
\[
  \beta^*(X) := \inf_{\beta\ge1} \sup_{\Delta>0}
  \Big\{ \text{$X$ admits a $(\beta, \Delta)$-Lipschitz decomposition} \Big\}, 
\]
and we extend this to a family of metric spaces $\cX$,
by defining $\beta^*(\cX) := \sup_{X \in \cX} \beta^*(X)$. 

Obtaining bounds on the decomposition parameter of various metrics
(and families of metrics) is of significant algorithmic importance,
and we list in \Cref{table:(Lipschitz-Decompositions)} several known bounds.
One fundamental example where we know of (nearly) tight bounds
is the metric space $\ell_p^d$, for $p \geq 1$,
which stands for $\RR^d$ equipped with the $\ell_p$ norm. 
For $p\in[1,2]$, we have $\beta^*(\ell_{p}^{d})=\Theta(d^{1/p})$
due to~\cite{CCGGP98},
and for $p\in[2,\infty]$ we have $\beta^*(\ell_{p}^{d}) = \tilde{\Theta}(d^{1/2})$
due to~\cite{Naor17}
(see discussion therein about an incorrect claim made in~\cite{CCGGP98}).%
\footnote{Throughout, the notation $\tilde{O}(f)$ hides $\poly(\log f)$ factors,
and $O_\alpha(\cdot)$ hides a factor that depends only on $\alpha$.
}
Observe that an upper bound for $X=\ell_p^d$ immediately extends to all subsets of it,
implying in particular a bound for the family $\cX$ of all finite subsets of $\ell_p^d$.
These bounds depend on $d$, and are thus most suitable for low-dimensional settings.

We focus on finite metrics $X$, aiming to bound $\beta^*(X)$ in terms of $n=|X|$,
which is often useful in high-dimensional settings.
For instance, it is well-known that $\beta^*= \Theta(\log{n})$
the family of all $n$-point metric spaces \cite{Bartal96}.
To write this assertion more formally,
define $\beta^*_n(X) := \beta^*(\{X'\subseteq X:\ |X|=n\})$
and then the above asserts that $\beta^*_n(\ell_\infty) =\Theta(\log n)$,
where we used that every finite metric embeds isometrically in $\ell_\infty$. 
For the family of $n$-point $\ell_{2}$ metrics,
combining $\beta^*(\ell_{2}^{d}) = \tilde{\Theta}(\sqrt{d})$
with the famous JL Lemma~\cite{JL84} immediately yields $\beta^{*}_{n}(\ell_{2}) = O(\sqrt{\log n})$,
which is tight by~\cite{CCGGP98}. 
For $n$-point $\ell_{p}$ metrics, $1<p < 2$,
we have $\beta^{*}_{n}(\ell_{p}) = \frac{O(\log^{1/p}{n})}{p-1}$ due to \cite{LN03draft, Naor17},
nearly matching the lower bound of $\beta^{*}_{n}(\ell_{p}) = \Omega(\log^{1/p}{n})$ from \cite{CCGGP98}.
However, for $n$-point $\ell_p$ metrics, $p > 2$, to the best of our knowledge,
the only known upper bound is $\beta^{*}_n(\ell_p) = O(\log n)$,
obtained by trivially applying the results for general $n$-point metric spaces. 
The following question was raised by Naor~\cite[Question 1]{Naor17},
see also~\cite[Question 83]{Naor24}.

\begin{question}[\cite{Naor17}]
Is it true that for every $p \in (2, \infty)$, $\beta^{*}_{n}(\ell_{p}) = o(\log{n})$? More ambitiously, is it true that $\beta^{*}_{n}(\ell_{p}) = O_p(\sqrt{\log{n}})$? 
\end{question}

Our main result, in \cref{thm:(Lipschitz-Decompositions-in-lp)},
answers the first part of this question in the affirmative.
Additionally, we show in \cref{sec:(Stochastic-Decomposition)}
an analogous result for another notion of decomposability 
that was introduced in~\cite{FN22} (and we call capped decomposition)
and is particularly suited for high-dimensional geometric spanners.

\begin{table}[]
\begin{tabular}{|l|l|l|l|}
\hline
Family of Metrics                               & $\beta^{*}$ or $\beta^{*}_n$                              & Reference                                    & 
                                                                                                                                                      Comments
  \\ \hline\hline
  $\ell_{p}^{d}$ spaces $1 \leq p \leq 2$                                             & $\Theta(d^{1/p})$                 & \cite{CCGGP98}                                 &
  \\ \hline
  $\ell_{p}^{d}$ spaces $p \geq 2$                                                 & $\Tilde{\Theta}(\sqrt{d})$               & \cite{Naor17}              &  
  \\ \hline\hline
  finite metrics                                                         & $\Theta(\log{n})$              & \cite{Bartal96}                               &

  \\ \hline
  $\ell_2$ space (Euclidean)                                                         & $\Theta(\sqrt{\log{n}})$                  & \cite{CCGGP98}                                &
  \\ \hline
  $\ell_p$ spaces $1 \leq p \leq 2$                                             & $\Theta_{p}(\log^{1/p}{n})$                 & \cite{LN03draft, Naor17}                                 &
  \\ \hline
  $\ell_p$ spaces $p \geq 2$                                                 & $O(\log^{1-1/p}{n})$               & \cref{thm:(Lipschitz-Decompositions-in-lp)}                                &  
    conjectured $\beta^{*}_{n}=\Theta(\sqrt{\log{n}})$
  \\ \hline\hline
  doubling constant $\lambda$                                             & $\Theta(\log{\lambda})$                   & \cite{GKL03}                                    &
  \\ \hline
  $K_r$-minor-free graphs                                                    & $O(r)$                  & \cite{AGGNT14,Filtser19}                                  & 
    conjectured $\beta^{*}=\Theta(\log{r})$
  \\ \hline
  graphs with genus $g$                                                   & $\Theta(\log{g})$                         & \cite{LS10, AGGNT14}                                &
  \\ \hline
  graphs with treewidth $w$                                                   & $\Theta(\log{w})$                 & \cite{FFIKLMZ24}                                  &
  \\ \hline
\end{tabular}
\caption{Known bounds on the decomposition parameter of some important families of metrics.}
\label{table:(Lipschitz-Decompositions)}
\end{table}

\subparagraph*{Geometric Spanners. }
A \emph{spanner with stretch $t\ge 1$} (in short a \emph{$t$-spanner})
for a finite metric \( M = (X, \rho) \) is a graph \( G = (X, E) \),
that satisfies \( \rho(x, y) \leq \rho_G(x, y) \leq t \cdot \rho(x, y) \)
for all \( x, y \in X \),
meaning that the shortest-path distance \( \rho_G \) in the graph \( G \)
approximates the original distance \( \rho(x, y) \) within factor \( t \),
where by definition every edge \( \{u, v\} \in E \) has weight \( \rho(u, v) \).
Of particular interest are spanners that are \emph{sparse},
meaning they contain a small number of edges, ideally linear in $n=|X|$.
Another important parameter is the \emph{lightness} of a spanner,
defined as the total weight of its edges divided by the weight of a minimum spanning tree of $X$.
Clearly, the lightness is at least $1$. 
These spanners are called geometric because the input is a metric space
(rather than a graph).
They are natural and useful representations of a metric, and as such,
have been studied extensively, see the surveys \cite{EPP00, Zwick01, ahmed2020graph}.
Spanners for $n$-point metrics in low-dimensional spaces
(e.g., in fixed-dimensional Euclidean space or doubling metrics)
are well-studied and well-understood.
For instance, metrics with doubling dimension $\ddim$
admit $(1+\varepsilon)$-spanners
with near-optimal sparsity \( n (1/\varepsilon)^{O(\ddim)} \)
and lightness \( (1/\varepsilon)^{O(\ddim)} \)~\cite{BLW19, LS23}.

However, in high-dimensional spaces, our understanding of spanners is rather limited.
Har-Peled, Indyk, and Sidiropoulos~\cite{HIS13} showed that
every \( n \)-point Euclidean metric admits,
an \( O(t) \)-spanner with \( \tilde{O}(n^{1+1/t^2}) \) edges for every $t\geq1$.
Filtser and Neiman~\cite{FN22} extended this result to all metric spaces
that admit a certain decomposition that we call capped decomposition (\cref{def:capped}),
showing that in those spaces, it is possible to construct spanners that are both sparse and light.
In particular, they showed that every \( n \)-point subset of \( \ell_{p} \), \( 1 < p \leq 2 \),
has an \( O(t) \)-spanner with \( n^{1+\tilde{O}(1/t^p)} \) edges and lightness \( n^{\tilde{O}(1/t^p)} \) for every $t\ge1$.
It remained open whether the spaces $\ell_p$ for \( p \in (2, \infty) \)
admit the aforementioned capped decomposition.
To the best of our knowledge, all known spanners for these spaces
have a tradeoff of stretch $O(t)$ with sparsity $O(n^{1+1/t})$.

\subsection{Our Results}
\label{sec:results}

Our main contribution is the construction of a Lipschitz decomposition
for finite \(\ell_p\) metrics, \(p \geq 2\), as follows. 

\begin{theorem}\label{thm:(Lipschitz-Decompositions-in-lp)}
Let $p \in [2,\infty]$. 
Then $\beta^*_n(\ell_p) = O(\log^{1-1/p}{n})$. 
That is, for every $n$-point metric $X\subset \ell_{p}$ and $\Delta > 0$, there exists an $(O(\log^{1-1/p}{n}), \Delta)$-Lipschitz decomposition of $X$.
\end{theorem}

Previously, this bound was known only for the extreme values $p=2,\infty$,
and in these cases it is actually tight. 
More precisely, for \( p = 2 \) our bound coincides with the well-known result
\( \beta^{*}_{n}(\ell_2) = \Theta(\sqrt{\log{n}}) \)~\cite{CCGGP98},
and for \( p = \Omega(\log{n}) \)
it is known that \( \beta^{*}_{n}(\ell_p) = \Theta(\log{n}) \),
because all $n$-point metrics embed into $\ell_p$ with $O(1)$-distortion~\cite{Matousek97}.
For intermediate values, say fixed \( p \in (2, \infty) \),
our bound is the first one to improve over $O(\log n)$,
which applies to all $n$-point metrics, 
and leaves a gap from the \( \Omega(\sqrt{\log{n}}) \) lower bound
that follows from Dvoretzky's Theorem~\cite{Dvoretzky61}. 
We compare our bound with those for other metric spaces 
in \cref{table:(Lipschitz-Decompositions)}.

The proof of \cref{thm:(Lipschitz-Decompositions-in-lp)} 
appears in \cref{sec:Lipschitz-Decomposition-p-greater-than-2},
and has interesting technical features. 
It relies on \emph{two} known decompositions of finite metrics,
one for general metrics and one for Euclidean metrics, 
that are composed via a metric-embedding tool called the Mazur map. 
Our decomposition method is data-dependent, i.e., not oblivious to the data,
and we discuss this intriguing aspect
in \cref{sec:Lipschitz-Decomposition-p-greater-than-2,sec:conclusion}.

\subparagraph*{Geometric Spanners for $p\ge 2$. }

We then use similar ideas %
to obtain a new bound for another notion of decomposability,
that was introduced in~\cite{FN22} and we call capped decomposition;
and this immediately yields geometric spanners in $\ell_p$, for $p\ge 2$.
While for $p=2$ these spanners coincide with the known bounds from~\cite{HIS13, FN22},
for fixed $2<p<\infty$, our spanners are the first improvement
over the trivial bounds that hold for all metric spaces. 

\begin{theorem}
\label{thm:(Spanner-In-l_p-p-greater-than-2)}
Let \( p \in [2, \infty) \) and \( t \geq 1 \).
Then every \( n \)-point metric \( X \subset \ell_{p} \) admits an \( O(t) \)-spanner of size \( \tO\left(n^{1 + 1/t^{q}}\right) \) and lightness \( \tO\left(n^{1/t^{q}}\right) \),
where \( q \in (1, 2) \) is such that \( \frac{1}{p} + \frac{1}{q} = 1 \).
\end{theorem}
The proof of this theorem appears in \cref{sec:capped-decomposition},
and includes both the spanner construction, which follows~\cite{FN22},
and our new bound for capped decomposition, which is the main technical result.

\subparagraph*{Geometric Spanners for $p\le 2$. }
We also sharpen the known spanner results for \(\ell_p\) spaces with $1<p<2$, 
which say that every \(n\)-point subset admits an \(O(t)\)-spanner 
with \(n^{1 + O(\log^2{t}/t^p)}\) edges and lightness \(n^{O(\log^2{t}/t^p)}\) for every \(t \geq 1\)~\cite{FN22}. 
We improve upon this result by eliminating the \(\log^2{t}\) factor in the exponent.

\begin{theorem}
\label{thm:(Spanner-In-l_p-p-less-than-2)}
Let \( p \in (1, 2] \) and \( t \geq 1 \).
Then every \( n \)-point metric \( X \subset \ell_{p} \) admits
an \( O(t) \)-spanner of size \( \tO(n^{1 + 1/t^{p}}) \)
and lightness \( \tO(n^{1/t^{p}}) \).
\end{theorem}
The proof of this theorem, presented in \cref{sec:(Geo_Spanners_High_Dim)}, follows the construction of \cite{FN22},
but replacing a key step, in which they rely on results from \cite{Ngu13}, 
with results from~\cite{AndoniIndyk}. 
Interestingly, our improved spanner bound ``matches'' the bounds of \cref{thm:(Spanner-In-l_p-p-greater-than-2)},
up to duality between $p$ and $q$.

\subparagraph*{Distance Labeling Schemes. }  
Distance labeling for a metric space $(X, \rho)$
assigns to each point $x\in X$ a label $l(x)$,
so that one can later recover (perhaps approximately)
the distance between any two points in $X$
based only on their labels (without knowledge of the metric space).
It was formulated in~\cite{Peleg00},
motivated by applications in distributed computing,
and has been studied intensively, see e.g.~\cite{GPPR04, FGK24}.
An immediate corollary of our main result in \cref{thm:(Lipschitz-Decompositions-in-lp)}
is a distance labeling scheme for finite metrics in \(\ell_p\) for \(p > 2\),
as follows.

\begin{theorem}
\label{thm:(Distance-Labeling-In-l_p-p_greater-than-2)}
Let \( p \in (2, \infty) \).
Then the family of \( n \)-point metrics in \( \ell_{p} \)
with pairwise distances in the range $[1,\Delta_{\text{max}}]$
admits a distance labeling scheme with approximation $O(\log^{1/q}{n})$
and label size $O(\log{n}\log{\Delta_{\text{max}}})$ bits,
where \( q \in (1, 2) \) is such that \( \frac{1}{p} + \frac{1}{q} = 1 \).
\end{theorem}
A formal definition of the distance labeling model
and a proof of \cref{thm:(Distance-Labeling-In-l_p-p_greater-than-2)}
appear in \cref{sec:(Distance-Labeling)}.

\subsection{Related Work}
\label{sec:related}

We focus on Lipschitz decomposition and on capped decomposition,
that was introduced in \cite{FN22}, 
but the literature studies several different decompositions of metric spaces
into low-diameter clusters, see e.g.~\cite{MN07, Filtser24}.
In particular, the notion of padded decomposition~\cite{Rao99, KL07}
is closely related and was used extensively,
see for example~\cite{Rao99, Bartal04, LN03draft, MN07, KLMN05}.
While a Lipschitz decomposition guarantees that nearby points are likely to be clustered together,
a padded decomposition guarantees that each point is, with good probability,
together with all its nearby points in the same cluster.
Remarkably, if a metric space admits a padded decomposition then it admits
also a Lipschitz decomposition with almost the same parameters~\cite{LN03draft},
however the other direction is not true, as demonstrated by $\ell_2^d$.
The problem of computing efficiently the optimal decomposition parameters
for an input metric space $(X,\rho)$ was studied in~\cite{KR11}.
Specifically for Lipschitz decomposition, %
they show that $\beta^*(X)$ can be $O(1)$-approximated in polynomial time (in $n$).

\section{Decompositions and Spanners in $\ell_p$ for $p>2$}
\label{sec:(Stochastic-Decomposition)}

In this section we consider finite subsets of \( \ell_{p} \) for \( p \in (2, \infty) \). 
We first present (in \cref{sec:Lipschitz-Decomposition-p-greater-than-2})
a new Lipschitz decomposition,
which proves \cref{thm:(Lipschitz-Decompositions-in-lp)}.
Next, we show (in \cref{sec:capped-decomposition})
a new construction of capped decomposition,
which is a related notion of decomposability
that was introduced in~\cite{FN22} without a concrete name. 
Finally we obtain (in \cref{sec:spanner}) new spanners, 
which prove \cref{thm:(Decomposable-Metrics-Spanner)}.
This is actually an immediate corollary of our capped decomposition,
by following the spanner construction of~\cite{FN22}.

\subsection{Lipschitz Decomposition in $\ell_{p}$ for $p\in(2,\infty)$}
\label{sec:Lipschitz-Decomposition-p-greater-than-2}

Before presenting the proof of \cref{thm:(Lipschitz-Decompositions-in-lp)}, 
we first provide the intuition behind the proof.
A common approach in many algorithms for metric spaces
is to embed the given metric into a simpler one (e.g., a tree metric),
solve the problem in the target metric,
and then pull back this solution to the original metric.
For our purpose,
of constructing a Lipschitz decomposition of $X \subset \ell_p$, $p > 2$,
a natural idea is to seek a low-distortion embedding of $X$ into $\ell_2$,
because we already have decompositions for that space,
namely, $\beta^{*}_{n}(\ell_2) = O(\sqrt{\log{n}})$.
Ideally, the embedding into $\ell_2$ would be \emph{oblivious}, 
meaning that it embeds the entire $\ell_p$ (not only $X$) into $\ell_2$,
but unfortunately such an embedding does not exist
(it would imply oblivious dimension reduction in $\ell_p$ for $p>2$,
which is provably impossible~\cite{CS02}).
We get around this limitation by employing a \emph{data-dependent} approach,
where the decomposition depends on the input set $X$.
More precisely, we use Mazur maps,
which provide a low-distortion embedding from $\ell_p$ to $\ell_2$,
but only for sets of bounded diameter (see \cref{cor:(low-dist-l_p-to-l_2-embedding)}).
We thus first decompose $X$ into bounded-diameter subsets
by applying a standard Lipschitz decomposition
(that is applicable for every $n$-point metric). 
The final decomposition is obtained by pulling back
the solution (clusters) we found in $\ell_2$.

We proceed to introduce some technical results
needed for our proof of \cref{thm:(Lipschitz-Decompositions-in-lp)}.
The first one is a well-known bound for Lipschitz decomposition of a finite metric.   

\begin{theorem} [\cite{Bartal96}]
\label{thm:(Bartal-Prob-Partitions)}
Every $n$-point metric \((X, \rho)\) admits an \(\left(O\left(\log n\right), \Delta\right)\)-Lipschitz decomposition for every \(\Delta > 0\).
\end{theorem}

Next, we define the \emph{Mazur map}, 
which is an explicit embedding $M_{p,q}:\ell_{p}^{m} \to l_q^m$ for $1<q<p<\infty$.
The image of an input vector $v$ is computed in each coordinate separately,
by raising the absolute value to power $p/q$ while keeping the original sign.
The next theorem appears in~\cite{BG19},
where it is stated as an adaptation of~\cite{benyamini1998geometric},
and we will actually need the immediate corollary that follows it.

\begin{theorem} [\cite{benyamini1998geometric,BG19}]
Let $1\leq{q}<p<\infty$ and $C_0>0$,
and let $M$ be the Mazur map $M_{p,q}$ scaled down by factor $\frac{p}{q}{C_0}^{p/q-1}$.
Then for all $x,y\in\ell_{p}$ such that $||x||_p,||y||_p\leq C_0$,
\[
  \tfrac{q}{p} (2 C_0 )^{1 - p/q} ||x - y||_{p}^{p/q}
  \leq ||M(x) - M(y)||_{q}
  \leq ||x - y||_{p} .
\]
\end{theorem}

\begin{corollary}
\label{cor:(low-dist-l_p-to-l_2-embedding)}
Let $2<p<\infty$.
Every $n$-point set \(X \subset \ell_{p}\) with diameter at most \( C_0>0\)
admits an embedding \(f: X \to \ell_{2}\) such that
\[
  \forall x,y \in X,
  \qquad
  \tfrac{2}{p} (2 C_0)^{1 - p/2} \|x - y\|_{p}^{p/2}
  \leq \|f(x) - f(y)\|_{2}
  \leq \|x - y\|_{p} .
\]
\end{corollary}

\begin{proof}[Proof of \cref{thm:(Lipschitz-Decompositions-in-lp)}]
Let \( \Delta > 0 \), and let \( X \subset \ell_{p} \) be an \( n \)-point metric space for \( p \in (2, \infty) \).
Construct a partition of $X$ in the following steps:
\begin{enumerate} \compactify
\item
  Construct for \( X \) an \( (O(\log{n}), \log^{1/p}{n} \cdot \Delta/4) \)-Lipschitz decomposition \( \Pinit  = \{ K_{1}, \ldots, K_{t} \} \)
  using \cref{thm:(Bartal-Prob-Partitions)}.
\item
  Embed each cluster $K_i \subset \ell_p$ into \( \ell_2 \)
  using the embedding \( f^{K_{i}} \) provided by \cref{cor:(low-dist-l_p-to-l_2-embedding)} for $C_0 := \log^{1/p}{n} \cdot \Delta/4$.
\item
  For each embedded cluster \( f^{K_{i}}(K_{i}) \),
  construct an \( (O(\sqrt{\log{n}}), \frac12 \Delta/ \log^{1/2 -1/p} n) \)-Lipschitz decomposition \( P_{i} = \{K_{i}^{1}, \ldots,K_{i}^{k_{i}} \} \) 
  using \cite{CCGGP98} and the JL Lemma~\cite{JL84}.
\item 
  The final decomposition $\Pout$ is obtained by taking the preimage of every cluster of every $P_i$. 
\end{enumerate}

It is easy to see that that $\Pout$ is indeed a partition of $X$,
consisting of $\sum_{i=1}^t k_i$ clusters.
Next, consider $x,y\in X$ and let us bound $\Pr[\Pout(x) \neq \Pout(y)]$. 
Observe that a pair of points can be separated only in steps 1 or 3. 
Therefore, 
\begin{align*}
  \Pr &\Big[\Pout(x) \neq \Pout(y)\Big]
  \\
  & \leq \Pr\Big[\Pinit (x) \neq \Pinit (y)\Big]
  +  \Pr\Big[P_i(f^{K_i}(x)) \neq P_i(f^{K_i}(y)) \mid \Pinit (x) = \Pinit (y) = K_{i}\Big]
  \\
  & \leq O(\log{n})\frac{\|x - y\|_{p}}{\log^{1/p}{n} \cdot \Delta/4} + O(\sqrt{\log{n}}) \frac{\|f^{K_i}(x) - f^{K_i}(y)\|_{2}}{\frac12 \Delta / \log^{1/2-1/p} n }
  \\
  & \leq O(\log^{1-1/p}{n}) \frac{\|x - y\|_{p}}{\Delta} ,
\end{align*}
where the last inequality follows because each $f^{K_i}$ 
is non-expanding on its cluster $K_{i}\subset \ell_p$.

It remains to show that the final clusters all have diameter at most \( \Delta \).
Let \( x, y \in X \) be in the same cluster, i.e., \( \Pout(x) = \Pout(y) \).
Then $\Pinit (x)=\Pinit (y)=K_i$ and $P_{i}(f^{K_i}(x))=P_{i}(f^{K_i}(y))$.
Combining the maximum possible diameter of \( \Pinit (x) \) and \( P_{i}(f^{K_i}(x)) \) with the contraction guarantees of \( f = f^{K_i} \), we get
\[
  \frac{2}{p} \Big( 2(\log^{1/p}{n})\frac{\Delta}{4} \Big)^{1 - p/2} \|x - y\|_{p}^{p/2} 
  \leq \|f(x) - f(y)\|_{2} 
  \leq \frac{\Delta}{2}\log^{1/p-1/2}{n}.
\]
Rearranging this, we obtain
\( \|x - y\|_p \leq \frac{\sqrt[2/p]{p/2}}{2}\Delta \leq \Delta \),
which completes the proof.
\end{proof}

\subsection{Capped Decomposition in $\ell_{p}$ for $p\in(2,\infty)$}
\label{sec:capped-decomposition}

We now present our construction of capped decomposition,
which is a notion of decomposability
that was introduced in~\cite{FN22} without a concrete name.
We start with its definition, and then present our construction.

\begin{definition}
\label{def:capped}
Let $(X, \rho)$ be a metric space. 
A distribution $\mathcal{D}$ over partitions of $X$ 
is called \emph{$(t, \Delta, \eta)$-capped} if
\begin{enumerate} \compactify
    \item for every partition $P \in \supp(\mathcal{D})$, all clusters $C \in P$ have $\diam(C) \leq \Delta$; and
    \item for every $x, y \in X$ such that \( \rho(x, y) \leq \frac{\Delta}{t} \),
    \[
    \Pr_{P\in\mathcal{D}}[P(x) = P(y)] \geq \eta.
    \]
\end{enumerate}
\end{definition}

Observe that here, unlike in Lipschitz decomposition, 
we have a guarantee on the probability that points \( x, y \in X \) are clustered together 
only if they are within distance \( \frac{\Delta}{t} \) of each other, 
hence the name "capped decomposition".
Moreover, the probability bound does not depend on the exact value of \( \rho(x,y) \). 
We say that $(X,\rho)$ admits a \( (t, \eta) \)-capped decomposition, where \( \eta = \eta(|X|, t) \), if it admits a \( (t, \Delta, \eta) \)-capped decomposition for every \( \Delta > 0 \). 
A family of metrics admits a \( (t, \eta) \)-capped decomposition 
if every metric in the family admits a \( (t, \eta) \)-capped decomposition.

\begin{theorem} 
\label{thm:(Decomposability-Of-lp)}
Let \( p \in (2,\infty) \).
Then every $n$-point metric in \( \ell_{p} \) admits
a \( (t, n^{-O(1/t^q)}) \)-capped decomposition for all $t \geq 1$,
where $q \in (1,2)$ is such that $\frac{1}{p}+\frac{1}{q}=1$.
\end{theorem}

Previously, such a decomposition was known only for the extreme case $p=2$
by \cite{FN22}, see \cref{prop:(Decomposability-Of-l2)}, 
and our bound above in fact converges to their bound when $p\to 2$. 
Our proof of \cref{thm:(Decomposability-Of-lp)}
is similar to \cref{thm:(Lipschitz-Decompositions-in-lp)},
and relies on two known capped decompositions, that we introduce next, 
together with the Mazur map \cref{cor:(low-dist-l_p-to-l_2-embedding)}. 

\begin{proposition}[\cite{FN22}]
\label{prop:(Decomposability-Of-l2)}
Every $n$-point subset of \( \ell_2 \) admits a \( (t, n^{-O(1/t^{2})}) \)-capped decomposition for all $t \geq 1$.
\end{proposition}

\begin{proposition}[Implicit in~\cite{MN07}]
\label{prop:(Decomposability-Of-General-Metrics)}
Every $n$-point metric space admits a \( (t, n^{-O(1/t)}) \)-capped decomposition for all $t \geq 1$.
\end{proposition}

\begin{proof}[Proof of \cref{thm:(Decomposability-Of-lp)}]
Let \(\Delta > 0\) and \(t \geq 1\).
Let \(X \subset \ell_p\) be an \(n\)-point subset of \(p \in (2, \infty)\),
where \(q\) is such that \(\frac{1}{p} + \frac{1}{q} = 1\).
Construct a partition of $X$ in the following steps:
\begin{enumerate} \compactify
\item
  Construct for $X$ a \( (t_{1} := t^{q}/4, \Delta_{1} := \Delta/4t^{1-q}, n^{-O(1/t^{q})}) \)-capped decomposition \(  \Pinit = \{ K_{1}, \ldots, K_{t} \} \)
  using \cref{prop:(Decomposability-Of-General-Metrics)}.
\item
  Embed each cluster $K_i \subset \ell_p$ into \( \ell_2 \)
  using the embedding \( f^{K_{i}} \) provided by \cref{cor:(low-dist-l_p-to-l_2-embedding)} for $C_0 := \Delta_{1}$.
\item
  For each embedded cluster \( f^{K_{i}}(K_{i}) \)
  construct a \( (t_{2} := t^{q/2}/2, \Delta_{2} := \Delta/2t^{1-q/2}, n^{-O(1/t^{q})}) \)-capped decomposition \( P_{i} = \{ K_{i}^{1}, \ldots, K_{i}^{k_{i}} \}\)
  using \cref{prop:(Decomposability-Of-l2)}.
\item 
  The final decomposition $\Pout$ is obtained by taking the preimage of every cluster of every $P_i$. 
\end{enumerate}

It is easy to see that that $\Pout$ is indeed a partition of $X$,
consisting of $\sum_{i=1}^t k_i$ clusters.
Next, consider $x,y\in X$ with \( \|x - y\|_p \leq \Delta/t \)
and let us bound $\Pr[\Pout(x) = \Pout(y)]$.
Observe that $\Delta_{1}/t_{1} = \Delta_{2}/t_{2} = \Delta/t$,
and therefore
\begin{align*}
  \Pr & \Big[ \Pout(x) = \Pout(y) \Big]
  \\
  & = \Pr \Big[ \Pinit(x) = \Pinit(y) \Big]
    \cdot \Pr\Big[ P_{i}(f^{K_{i}}(x)) = P_{i}(f^{K_{i}}(y)) \mid \Pinit(x) = \Pinit(y) = K_{i} \Big]
  \\
  & \geq n^{-O\left(1/t^{q}\right)} \cdot n^{-O\left(1/t^{q}\right)} = n^{-O\left(1/t^{q}\right)} ,
\end{align*}
where the inequality follows because each $f^{K_i}$
is non-expanding on its cluster $K_i \subset \ell_p$. 

It remains to show that each cluster has diameter at most \( \Delta \).
Let \( x, y \in X \) be in the same cluster, i.e., \( \Pout(x) = \Pout(y) \).
Then $\Pinit (x)=\Pinit (y)=K_i$ and $P_{i}(f^{K_i}(x))=P_{i}(f^{K_i}(y))$.
Combining the maximum possible diameter of \( \Pinit (x) \) and \( P_{i}(f^{K_i}(x)) \) with the contraction guarantees of \( f = f^{K_i} \), we get
\[
  \frac{2}{p} \Big( 2\frac{\Delta}{4t^{1-q}} \Big)^{1 - p/2} \|x - y\|_{p}^{p/2} 
  \leq \|f(x) - f(y)\|_{2} 
  \leq \frac{\Delta}{2t^{1-q/2}}.
\]
Rearranging this, we obtain
\( \|x - y\|_p \leq \frac{\sqrt[2/p]{p/2}}{2}\Delta \leq \Delta \),
which completes the proof.
\end{proof}

\subsection{Spanners in $\ell_{p}$ for $p\in(2,\infty)$}
\label{sec:spanner}

We can now prove \cref{thm:(Spanner-In-l_p-p-greater-than-2)},
by applying the following spanner construction of~\cite{FN22}.

\begin{theorem} [\cite{FN22}]
\label{thm:(Decomposable-Metrics-Spanner)}
Let \( (X, \rho) \) be an \( n \)-point metric space admitting a \( (t, \eta) \)-capped decomposition for some $t \geq 1$.
Then, for every \( \epsilon \in (0, 1/8) \),
there exists a \( (2+\epsilon)t \)-spanner for \( X \)
with \( O_{\epsilon} (\frac{n}{\eta} \cdot \log n \cdot \log t ) \) edges
and lightness \( O_{\epsilon} (\frac{t}{\eta} \cdot \log^2 n ) \).
\end{theorem}

\begin{proof}[Proof of \cref{thm:(Spanner-In-l_p-p-greater-than-2)}]
The proof follows directly by combining \cref{thm:(Decomposability-Of-lp)} and \cref{thm:(Decomposable-Metrics-Spanner)}, 
as we can assume \( t = O(\log{n}) \) without loss of generality.
\end{proof}

\section{Spanners in $\ell_{p}$ for $p \in (1,2)$}
\label{sec:(Geo_Spanners_High_Dim)}

This section presents an improved construction of geometric spanners in \(\ell_p\) for \(p \in (1, 2)\). 
Previously, \(O(t)\)-spanners of size \(O(n^{1+\log^2{t}/t^p})\) for all \(t \geq 1\) 
were constructed in~\cite{FN22};
in  particular, setting \(t = (\log{n} \log{\log{n}})^{1/p}\) yields 
an $O(t)$-spanner of near-linear size $\tO(n)$. 
We first present in \Cref{sec:Linear-Size-Spanners-in-l_p-p-less-than-2} 
two different constructions of near-linear-size spanners with a slightly better stretch.
Then in \Cref{sec:capped-decomposition-p-smaller-than-2} we use yet another technique,
namely Locality Sensitive Hashing (LSH), 
to slightly improve the construction of \cite{FN22} of spanners with general stretch $O(t)$.

\subsection{Spanners of Near-Linear Size}
\label{sec:Linear-Size-Spanners-in-l_p-p-less-than-2}

We slightly improve the near-linear size spanner construction of \cite{FN22}
by shaving the $(\log\log{n})^{1/p}$ factor from the stretch, as follows. 

\begin{theorem}
\label{thm:linear-size-spanners-in-l_p-p-less-than-two}
For every fixed $p\in(1,2)$, every $n$-point metric $X\subset \ell_{p}$
admits an $O(\log^{1/p}{n})$-spanner of size $\Tilde{O}(n)$. 
\end{theorem}

We present two related but different proofs for this theorem.
Both are based on modifying the spanner algorithm for $\ell_2$ from~\cite{HIS13}, 
and therefore we start with an overview of that algorithm.
Given an input set \(X \subseteq \ell_2\),
the algorithm begins by constructing a hierarchical set of \( 2^i \)-nets
$X = N_0 \supseteq N_1 \supseteq \cdots \supseteq N_{\log{\Delta_X}}$,
where we assume that the minimum and maximum distances in $X$
are $1$ and $\Delta_X$, respectively.
Then, for each level $i$, it constructs an \( (O(\sqrt{\log n}), O(\sqrt{\log n}) \cdot 2^{i+1}) \)-Lipschitz decomposition of $N_i$
by combining the JL Lemma~\cite{JL84} with the Lipschitz decomposition of~\cite{CCGGP98}.
For each cluster in it, the algorithm add to the spanner edges in a star-like fashion,
meaning that all cluster points are connected to one arbitrary point within the cluster.
The last two steps are repeated \( O(\log n) \) times
to ensure that with high probability, for each level $i$,
every \( x, y \in N_i \) with \( \|x - y\|_2 \leq 2^{i+1} \)
are clustered together in at least one of the $O(\log n)$ repetitions.
It is shown in~\cite{HIS13} that
this algorithm constructs an \( (O(\sqrt{\log{n}})) \)-spanner of size \( \Tilde{O}(n) \).

\begin{proof}
[Proof of \cref{thm:linear-size-spanners-in-l_p-p-less-than-two} via Lipschitz Decomposition]
Observe that the above algorithm of \cite{HIS13}
uses the fact that the points lie in $\ell_2$
only for the construction of Lipschitz Decompositions,
and relies on an optimal decomposition for finite $\ell_2$ metrics
to conclude that the spanner's stretch is $O(\beta^{*}_{n}(\ell_2))$. 
For finite \(\ell_p\) metrics, \(p \in (1, 2)\),
we can use instead a Lipschitz decomposition from~\cite{LN03draft, Naor17},
which has \(\beta = \frac{O(\log^{1/p}{n})}{p - 1}\),
to conclude the claimed stretch.
\end{proof}

We next present a proof that modifies the algorithm of \cite{HIS13} differently,
and relies on a decomposition that is similar to a Lipschitz decomposition
but has slightly weaker guarantees. 
Interestingly, this technique yields a slightly stronger result
than \Cref{thm:linear-size-spanners-in-l_p-p-less-than-two},
where \(p\) need not to be fixed and can depend on $n$ (e.g., $p\to 1$). 
We proceed to introduce some technical results from \cite{BG16}
regarding a weak form of dimensionality reduction in \(\ell_p\),
for \(p \in [1, 2]\), which are needed for our proof.

\begin{definition} [\cite{OR02}]
Let \( (X, \rho) \), \( (Y, \tau) \) be metric spaces
and \([a, b]\) be a real interval.
An embedding \( f : X \rightarrow Y \)
is called \emph{\([a, b]\)-range preserving with distortion $D \ge 1$ }
if there exists \( c > 0 \) such that for all \( x, x' \in X \):
\begin{enumerate} \compactify
    \item If \( a \leq \rho(x, x') \leq b \), then \( \rho(x, x') \leq c \cdot \tau(f(x), f(x')) \leq D \cdot \rho(x, x') \).
    \item If \( \rho(x, x') > b \), then \( c \cdot \tau(f(x), f(x')) \geq b \).
    \item If \( \rho(x, x') < a \), then \( c \cdot \tau(f(x), f(x')) \leq D \cdot a \).
\end{enumerate}    
\end{definition}
We say that \( (X,\rho) \) admits \emph{an $R$-range preserving embedding}
into \( (Y,\tau) \) with distortion \( D \),
if for all \( u > 0 \), there exists a \([u, uR]\)-range preserving embedding
into \( Y \) with distortion \( D \).   

\begin{theorem} [\cite{BG16}] 
\label{thm:weak-dimension-reduction}
Let \( 1 \leq p \leq 2 \). For every \( n \)-point set \( S \subset \ell_{p} \), and for every range parameter \( R > 1 \), there exists an \( R \)-range preserving embedding \( f : S \rightarrow \ell_{p}^{k} \) with distortion \( 1 + \epsilon \), such that $k = O\left(\frac{R^{O(1/\epsilon)} \cdot \log n}{\epsilon}\right)$.
\end{theorem}

\begin{proof}
[Proof of \cref{thm:linear-size-spanners-in-l_p-p-less-than-two} via Weak Dimension Reduction]
Observe that the above algorithm of \cite{HIS13}
only requires the decomposition of each net \(N_i\)
to ensure that points \(x, y \in N_i\) with \(\|x - y\|_2 \leq 2^{i+1}\)
are clustered together with constant probability, 
and that the diameter of all clusters is at most \( O(\sqrt{\log{n}}) \cdot 2^i\);
of course, for $X \subset \ell_p$, $p\in(1,2)$,
we replace the $O(\sqrt{\log{n}})$ factor with $O(\log^{1/p}{n})$.
A careful examination shows that these properties are preserved by first reducing the dimension using the range-preserving embedding provided by \Cref{thm:weak-dimension-reduction} with \(\varepsilon = \frac{1}{2}\) and \(R = 2\), and then constructing a Lipschitz decomposition for the image points in \(\ell_p^{O(\log{n})}\) using~\cite{CCGGP98}.
\end{proof}

\subsection{Spanners with Stretch-Size Tradeoff}
\label{sec:capped-decomposition-p-smaller-than-2}

We now present, in \Cref{thm:(Spanner-In-l_p-p-less-than-2)},
a construction of \(O(t)\)-spanners in $\ell_p$, where $p\in(1,2)$,
of size \(\Tilde{O}(n^{1+1/t^{p}})\) for all \(t \geq 1\),
which slightly improves over the \(O(t)\)-spanners of size \(\Tilde{O}(n^{1+\log^{2}{t}/t^{p}})\) from~\cite{FN22}.
It is worth noting that \Cref{thm:(Spanner-In-l_p-p-less-than-2)}
generalizes the results of \Cref{thm:linear-size-spanners-in-l_p-p-less-than-two},
and thus provides an alternative proof for it.

\begin{definition} [LSH~\cite{IM98}]
Let $\mathcal{H}$ be a family of hash functions mapping a metric $(X, \rho)$ to some universe $U$.
We say that $\mathcal{H}$ is \emph{$(r, tr, p_1, p_2)$-sensitive}
if for every $x, y \in X$, the following is satisfied:
\begin{enumerate} \compactify
    \item If $\rho(x, y) \leq r$, then $\Pr_{h \in \mathcal{H}}[h(x) = h(y)] \geq p_1$.
    \item If $\rho(x, y) > tr$, then $\Pr_{h \in \mathcal{H}}[h(x) = h(y)] \leq p_2$.
\end{enumerate}    
Such $\mathcal{H}$ is called an \emph{LSH family}
with parameter $\gamma := \frac{\log (1/p_1)}{\log (1/p_2)}$. 
\end{definition}

\begin{lemma} [\cite{FN22}]
\label{lem:LSH-implies-capped-decomposition}
    Let $(X, \rho)$ be a metric space such that for every $r > 0$, there exists a $(r, t r, p_1, p_2)$-sensitive LSH family with parameter $\gamma$. Then $(X, \rho)$ admits a $(t, n^{-\mathcal{O}(\gamma)})$-capped decomposition.
\end{lemma}

For $p=2$, the LSH family constructed in~\cite{AI06}
can be used in \cref{lem:LSH-implies-capped-decomposition}
to conclude that $\ell_2$ admits
a $(t, n^{-O(1/t^{2})})$-capped decomposition for every $t \geq 1$ \cite{FN22},
thereby proving \cref{thm:(Spanner-In-l_p-p-less-than-2)} for this case of $p=2$.
In a similar fashion,
an LSH family constructed in~\cite{Ngu13} for $p \in (1,2)$
was used in~\cite{FN22} to show that these spaces admit a $(t, n^{-O(\log^2{t}/t^{p})})$-capped decomposition.
We observe that this result can be improved by replacing the LSH family from~\cite{Ngu13},
with an alternative one that is briefly mentioned in~\cite{AndoniIndyk},
and consequently prove \cref{thm:(Spanner-In-l_p-p-less-than-2)}.
For completeness, we reproduce this LSH family for $\ell_p$, where $p\in(1,2)$. 

\begin{lemma} [\cite{AndoniIndyk}]
\label{lem:LSH-for-l_p}
    Let $p\in(1,2)$, $r>0$, and large enough $t>1$. Then there exists a $(r, t  r, p_1, p_2)$-sensitive LSH family for $\ell_p$ with parameter $\gamma=\frac{1}{t^p}+o(1)$.
\end{lemma}

\begin{proof}
Let \( p \in (1,2) \), \( r > 0 \), and sufficiently large \( t > 1 \).
Let \( f : \ell_p \to \ell_2 \) be the isometric embedding of the \( (p/2) \)-snowflake of \( \ell_p \) into \( \ell_2 \) from~\cite[Theorem 4.1]{Kal08}.
Take \( r' = r^{p/2} \) and \( t' = t^{p/2} \), and let \( \mathcal{H} \) be the \( (r', t' r', p_1, p_2) \)-sensitive LSH family for \( \ell_2 \) with parameter \( \gamma = \frac{1}{t'^2} + o(1) \) from~\cite{AI06}. Observe that, for every \( x, y \in \ell_p \), if \( \|x - y\|_p \leq r \), then \( \|f(x) - f(y)\|_2 = \|x - y\|_p^{p/2} \leq r^{p/2} = r' \), and thus 
\[
\Pr_{h \in \mathcal{H}}[h(f(x)) = h(f(y))] \geq p_1.
\]
Similarly, if \( \|x - y\|_p > t r \), then \( \|f(x) - f(y)\|_2 = \|x - y\|_p^{p/2} > (t r)^{p/2} = t' r' \), and hence 
\[
\Pr_{h \in \mathcal{H}}[h(f(x)) = h(f(y))] \leq p_2.
\]
We therefore conclude that \( \mathcal{H} \circ f \) is an \( (r, tr, p_1, p_2) \)-sensitive LSH family for \( \ell_p \) with parameter \( \gamma = \frac{1}{t^p} + o(1) \).
\end{proof}

\begin{proof}[Proof of \cref{thm:(Spanner-In-l_p-p-less-than-2)}]
The proof follows immediately by constructing a capped decomposition 
based on \cref{lem:LSH-implies-capped-decomposition} and \cref{lem:LSH-for-l_p},
and using it in the spanner construction from \cref{thm:(Decomposable-Metrics-Spanner)}.
\end{proof}

\begin{remark}
While \cite[Theorem 4.1]{Kal08} does not provide an efficiently computable embedding, 
one can compute such an embedding for a finite set of points in polynomial time
by \cite{LLR95}.
\end{remark}

\section{Distance Labeling}
\label{sec:(Distance-Labeling)}
In the distance labeling model,
a scheme is designed for an entire a family $\cX$ of $n$-point metrics
(and in some scenarios, all these metrics have the same point set $X$,
e.g., different graphs on the same vertex set).
A \emph{scheme} is an algorithm that preprocesses each metric $X$ in $\cX$
and assigns to each point \(x \in X\) a label \(l(x)\).

\begin{definition}
A scheme is \emph{a distance labeling}
with approximation $D\ge 1$ and label size of $k$ if 
\begin{enumerate} \compactify
\item
  every label (for every point in every metric in $\cX$)
  consists of at most $k$ bits; and 
\item
  there is an algorithm $\cA$ that, given the labels $l(x),l(y)$
  of two points $x,y$ in a metric $(X,\rho)\in \cX$
  (but not given $(X, \rho)$ or the points $x,y$),
  outputs an estimate $\cA(l(x), l(y))$ that satisfies 
  \[
    \rho(x, y) \leq \cA(l(x), l(y)) \leq D \cdot \rho(x, y) .
  \]
\end{enumerate}
\end{definition}

The following theorem was presented in \cite{GKL03} with limited details,
and we include a proof of it below for completeness.

\begin{theorem} [\cite{GKL03}]
\label{thm:(Distance-Labeling-In-Decomposable-Metrics)}
Let $\cX$ be a family of $n$-point metrics, 
and assume that all the pairwise distances in all metrics $(X,\rho)$ in $\cX$
are in the range $[1,\Delta_{\text{max}}]$.
Then $\cX$ admits a distance-labeling scheme
with approximation $O(\beta^*(\cX))$ and label size $O(\log n \log\Delta_{\max})$ bits. 
\end{theorem}

It is straightforward to see that \cref{thm:(Distance-Labeling-In-l_p-p_greater-than-2)}
follows by combining \cref{thm:(Distance-Labeling-In-Decomposable-Metrics)}
and \cref{thm:(Lipschitz-Decompositions-in-lp)}.

\begin{proof}[Proof of \cref{thm:(Distance-Labeling-In-Decomposable-Metrics)}]
We first describe the preprocessing algorithm,
denoting $\beta := \beta^*(\cX)$. 
Perform the following steps for all levels \( i = 0, \ldots, \log{\Delta_{\text{max}}} \).
Begin by constructing a \( (\beta, \Delta_i := 4\beta 2^{i}) \)-Lipschitz decomposition,
and observe that every two points \( x, y \in X \) with \( \rho(x,y) \leq 2^{i} \) are separated with probability at most \( \frac{1}{4} \).
Then, assign a random bit to each cluster,
and observe that if two points are at distance greater than \( \Delta_i \),
they always fall in different clusters,
hence, the probability that they are assigned the same bit is exactly \( \frac{1}{2} \), and if they are at distance at most $2^i = \Delta_i/(4\beta)$
they are assigned the same bit with probability at least \( \frac{3}{4} \).
Repeat the last two steps \( k = O(\log{n}) \) times,
and then with high probability, every two points \( x, y \) are assigned
the same bit at least \( \frac{5}{8}k \) times if \( \rho(x, y) \leq \Delta_i/(4\beta) \)
and fewer than \( \frac{5}{8}k \) times if \( \rho(x, y) > \Delta_i \).
Finally, label each point by concatenating the bit assigned to its cluster
in all the repetitions at all levels.

The label-size analysis is straightforward.
It remains to show that, given two labels \(l(x),l(y)\),
it is possible to approximate the distance \(\rho(x, y)\) within factor $O(\beta)$.
This can be achieved by identifying the smallest level $i$
such that $x$ and $y$ are assigned the same bit at least $\frac{5}{8}k$ times,
and then the above analysis (used in contrapositive form) implies that
$\Delta_{i-1}/(4\beta) < \rho(x,y) \leq \Delta_i$,
where by convention $\Delta_{-1}:=1$.
\end{proof}

\section{Future Directions}
\label{sec:conclusion}

\subparagraph*{Lipschitz Decompositions.} 

We stress that our decomposition in \Cref{thm:(Lipschitz-Decompositions-in-lp)}
employs a data-dependent approach, and is not oblivious to the input set $X$
(as, say, the decomposition for $\ell_2$ in \cite{CCGGP98},
even when it applied together with the JL Lemma).
In retrospect, this feature is perhaps not very surprising,
because data-dependent approaches have been already shown to be effective
for central problems, such as nearest neighbor search~\cite{ANNRW18, KNT21}.
We thus mention that a major open problem in the field
is whether dimension reduction is possible in $\ell_{p}$ for $p\neq 1,2,\infty$;
we know that for $p>2$ this is not possible via oblivious mappings~\cite{CS02}, 
raising the question whether data-dependent mappings can overcome this limitation.

Another possible direction is for $1<p<2$,
where the known bound $\beta^{*}_{n}(\ell_{p}) = \frac{O(\log^{1/p}{n})}{p-1}$
is achieved using an oblivious decomposition~\cite{LN03draft, Naor17}.
Motivated by the case $p\to 1$,
a tight bound $\beta^{*}_{n}(\ell_{p}) = \Theta(\log^{1/p}{n})$
is conjectured in~\cite[Conjecture 1]{Naor17}
(see also~\cite[Conjecture 82]{Naor24}).
We believe that a data-dependent decomposition could potentially play a key role
in proving this conjecture.

\subparagraph*{Geometric Spanners.}
The geometric spanners in \cite{HIS13,FN22} for $\ell_p$, $1<p\le 2$, 
are not known to be optimal, i.e., we do not know of matching lower bounds,
except for a more restricted case of $2$-hop spanners \cite{HIS13}. 
We conjecture that tight instances exist in these spaces,
i.e., the spanner bounds obtained in \cite{HIS13, FN22} are optimal for every stretch $t$. 
We similarly do not know of matching lower bounds
for the geometric spanners in $\ell_p$, for fixed $2\le p<\infty$, 
that we obtain in \cref{thm:(Spanner-In-l_p-p-greater-than-2)},
and it is quite plausible that our upper bounds are not tight. 
We do know however, based on known results, that for every $n$,
there exist tight instances in $\ell_p$ for $p=\Omega(\log{n})$.

\ifLIPICS
    \bibliographystyle{alphaurl}
    \bibliography{references.bib}
\else
    {\small
      \bibliographystyle{alphaurl}
      \bibliography{references}
    } %
\fi

\end{document}